\documentclass[12 pt]{amsart}
\usepackage{amssymb,latexsym,amsmath,amscd,amsthm,graphicx, color}
\usepackage[all]{xy}
\usepackage{pgf,tikz}
\usepackage{mathrsfs}
\usepackage{cite}
\usetikzlibrary{arrows}
\definecolor{uuuuuu}{rgb}{0.26666666666666666,0.26666666666666666,0.26666666666666666}
\definecolor{xdxdff}{rgb}{0.49019607843137253,0.49019607843137253,1.}
\definecolor{ffqqqq}{rgb}{1.,0.,0.}

\definecolor{uuuuuu}{rgb}{0.26666666666666666,0.26666666666666666,0.26666666666666666}
\definecolor{qqwuqq}{rgb}{0.,0.39215686274509803,0.}
\definecolor{zzttqq}{rgb}{0.6,0.2,0.}
\definecolor{xdxdff}{rgb}{0.49019607843137253,0.49019607843137253,1.}
\definecolor{qqqqff}{rgb}{0.,0.,1.}
\definecolor{cqcqcq}{rgb}{0.7529411764705882,0.7529411764705882,0.7529411764705882}

\theoremstyle{plain}

\newtheorem{theorem}[subsection]{Theorem}

\newtheorem{lemma}[subsection]{Lemma}

\newtheorem{prop}[subsection]{Proposition}

\theoremstyle{definition}
\newtheorem{remark}[subsection]{Remark}

\newtheorem{note}[subsection]{Note}

\newcommand{\uu}{\cup}
\newcommand{\ii}{\cap}

\newcommand{\sci}{\subset}
\newcommand{\es}{\emptyset}
\newcommand{\set}[1]{\{#1\}}

\newcommand{\ga}{\alpha}
\newcommand{\gb}{\beta}

\renewcommand{\gg}{\gamma}

\newcommand{\gl}{\lambda}

\newcommand{\gs}{\sigma}
\newcommand{\gt}{\tau}

\newcommand{\tit}{\textit}

\newcommand{\C}[1]{\mathcal{#1}}
\newcommand{\D}[1]{\mathbb{#1}}

\newcommand{\te}{\text}

\begin{document}

To appear, Journal of Interdisciplinary Mathematics
\title{Optimal quantization for nonuniform Cantor distributions}

\author{Lakshmi Roychowdhury}

\address{School of Mathematical and Statistical Sciences\\
University of Texas Rio Grande Valley\\
 1201 West University Drive\\
Edinburg, TX 78539, USA.}
\email{lakshmiroychowdhury82@gmail.com}

\thanks{This work was done in partial fulfillment of the author's Master's thesis at Texas A\&M University under the direction
of Professor S.N. Lahiri.}

\date{}
\maketitle

\pagestyle{myheadings}\markboth{Lakshmi Roychowdhury}{Optimal quantization for nonuniform Cantor distributions}

\begin{abstract}  Let $P$ be a Borel probability measure on $\mathbb R$ such that $P=\frac 1 4 P\circ S_1^{-1} +\frac 3 4 P\circ S_2^{-1}$, where $S_1$ and $S_2$ are two similarity mappings on $\mathbb R$ such that $S_1(x)=\frac 1 4 x $ and $S_2(x)=\frac 1 2 x +\frac 12$ for all $x\in \mathbb R$. Such a probability measure $P$ has support the Cantor set generated by $S_1$ and $S_2$. For this  probability measure, in this paper, we give an induction formula to determine the optimal sets of $n$-means and the $n$th quantization errors for all $n\geq 2$.  We have shown that the same induction formula also works for the Cantor distribution $P:=\psi^2 P\circ S_1^{-1} +\psi^4 P\circ S_2^{-1}$ supported by the Cantor set generated by $S_1(x)=\frac 13x$ and $S_2(x)=\frac 13 x+\frac 23$ for all $x\in \mathbb R$, where $\psi$ is the square root of the Golden ratio $\frac 12(\sqrt 5-1)$. In addition, we give a counter example to show that the induction formula does not work for all Cantor distributions. Using the induction formula we obtain some results and observations which are also given in this paper.
\end{abstract}

\section{Introduction}

Quantization of continuous random signals (or random variables and processes) is an
important part of digital representation of analog signals for various coding techniques
(e.g., source coding, data compression, archiving, restoration). The oldest example of quantization in statistics is rounding off. Sheppard (see \cite{S}) was the first who analyzed rounding off for estimating densities by histograms.
Any real number $x$ can be rounded off (or quantized) to the nearest
integer, say $q(x) = [x]$, with a resulting quantization error $e(x) = x - q(x)$, for
example, $q(2.14259) =2$. It means that the restored signal may differ
from the original one and some information can be lost.
 Thus, in quantization of a continuous set of values there is always a distortion (also known as noise or error) between the original set of values and the quantized set of values. One of the main goals in quantization theory is to find a set of quantizers for which the distortion is minimum. For the most comprehensive overview of quantization
one can see \cite{GN} (for later references, see \cite{GL}).  Over the years several authors estimated the distortion measures for quantizers (see, e.g., \cite{LCG} and \cite{Z}). A class of
asymptotically optimal quantizers with respect to an $r$th-mean error distortion measure is
considered in \cite{GL1} (see also \cite{CG, SS1}). A different approach for uniform scalar quantization is
developed in \cite{SS2}, where the correlation properties of a Gaussian process are exploited to
evaluate the asymptotic behavior of the random quantization rate for uniform quantizers.
General quantization problems for Gaussian processes in infinite-dimensional functional
spaces are considered in \cite{LP}. In estimating weighted integrals of time series with no quadratic mean derivatives, by means of samples at discrete times, it is known that the rate of convergence of mean-square error is reduced from $n^{-2}$ to $n^{-1.5}$ when the samples are quantized (see \cite{BC1}). For smoother time series, with $k=1,2, \cdots$ quadratic mean derivatives, the rate of convergence is reduced from $n^{-2k-2}$ to $n^{-2}$ when the samples are quantized, which is a very significant reduction (see \cite{BC2}). The interplay between sampling and quantization is also studied in \cite{BC2}, which asymptotically leads to optimal allocation between the number of samples and the number of levels of quantization. Quantization also seems to be a promising tool in recent development in numerical probability (see, e.g., \cite{PPP}).

Let $\D R^d$ denote the $d$-dimensional Euclidean space equipped with a metric $\|\cdot\|$ compatible with the Euclidean topology. Let $P$ be a Borel probability measure on $\D R^d$ and $\ga$ be a finite subset of $\D R^d$. Then, $\int \min_{a \in \ga} \|x-a\|^2 dP(x)$ is often referred to as the \tit{cost,} or \tit{distortion error} for $\ga$ with respect to the probability measure $P$, and is denoted by $V(P; \ga)$.  Write
$\C D_n:=\set{\ga \sci \D R^d : 1\leq \te{card}(\ga)\leq n}$.  Then, $\inf\set{V(P; \ga) : \ga \in \C D_n}$ is called the \tit{$n$th quantization error} for the probability measure $P$, and is denoted by $V_n:=V_n(P)$. A set $\ga$ for which the infimum occurs and contains no more than $n$ points is called an \tit{optimal set of $n$-means}. The set of all optimal sets of $n$-means for a probability measure $P$ is denoted by $\C C_n(P)$. Since $\int \| x\|^2 dP(x)<\infty$ such a set $\ga$ always exists (see \cite{ GKL, GL, GL1}). To know more details about quantization, one is referred to \cite{AW, GG, GL1, GN}. For any finite $\ga\sci \D R^d$, the \tit{Voronoi region} generated by an element $a \in \ga$ is defined by the set of all elements in $\D R^d$ which are closer to $a$ than to any other element in $\ga$, and is denoted by $M(a|\ga)$, i.e.,
\[M(a|\ga)=\set{x \in \D R^d : \|x-a\|=\min_{b \in \ga}\|x-b\|}.\]
A Borel measurable partition $\set{A_a : a \in \ga}$ of $\D R^d$  is called a \tit{Voronoi partition} of $\D R^d$ with respect to $\ga$ (and $P$) if $P$-almost surely $A_a \sci M(a|\ga) \te{ for every $a \in \ga$}.$
Let $P$ be a continuous Borel probability measure on $\D R^d$, $\ga$ be an optimal set of $n$-means for $P$, and $a\in \ga$. Then, it is well-known that (see \cite{GG, GL1}) $M(a|\ga)$ has positive measure with its boundary has measure zero, $a$ is the conditional expectation of its own Voronoi region, and  $P$-almost surely the set $\set{M(a|\ga) : a \in \ga}$ forms a Voronoi partition of $\D R^d$.

Let $S_1, S_2 : \mathbb R \to \mathbb R$ be two contractive similarity mappings such that
$S_1(x)=s_1 x$ and $ S_2 (x)=s_2 x +(1-s_2)$,
where $0<s_1, s_2<1$ and $s_1+s_2<1$.
Let $(p_1, p_2)$ be a probability vector with $p_1, p_2>0$. Then, there exists a unique Borel probability measure $P$
on $\mathbb R$ such that
$P=p_1 P\circ S_1^{-1}+p_2 P\circ S_2^{-1}$, where $P\circ S_i^{-1}$ denotes the image measure of $P$ with respect to
$S_i$ for $i=1, \,2$ (see \cite{H}). If $k\in \D N$, and $\sigma:=\sigma_1\sigma_2 \cdots \sigma_k \in \{ 1, 2\}^k$, then we call $\gs$ a word of length $k$, and denote it by
$|\gs|:=k$. By $\set{1, 2}^\ast$, we denote the set of all words including the empty word $\es$. Notice that the empty word has length zero. We write \[S_\gs:=S_{\gs_1}\circ \cdots \circ S_{\gs_k}, \, p_\gs:=p_{\gs_1} p_{\gs_2} \cdots p_{\gs_k}, \, s_\gs:=s_{\gs_1} s_{\gs_2} \cdots s_{\gs_k} \text{ and } J_\gs:=S_\gs([0, 1]).\]
Then, the set $C:=\bigcap_{k\in \mathbb N} \bigcup_{\sigma \in \{1, 2\}^k} J_\sigma$ is known as the \textit{Cantor set} generated by the
two mappings $S_1$ and $S_2$, and equals the support of the probability measure $P$, where $P$ can be written as
 \[P=\sum_{\gs\in I^k} p_\gs P\circ S_\gs^{-1}.\]
 $P$ is called a \tit{uniform Cantor distribution} supported by the Cantor set $C$ if $\frac{p_1}{s_1}=\frac{p_2}{s_2}$, otherwise, $P$ is called a \tit{nonuniform Cantor distribution}. When $s_1=s_2=\frac 13$ and
$p_1=p_2=\frac 12$, i.e., for the uniform distribution  $P=\frac 1 2 P\circ S_1^{-1}+\frac 1 2 P\circ S_2^{-1}$ with support the classical
Cantor set, Graf and Luschgy determined the optimal sets of $n$-means and the $n$th quantization errors for all $n\geq 2$ (see \cite{GL2}).  In this paper, we have taken $s_1=\frac 1 4$, $s_2=\frac 1 2$, $p_1=\frac 1 4$ and $p_2=\frac 3 4$, i.e., the probability measure $P$ considered here is nonuniform and satisfies
$P=\frac 1 4 P\circ S_1^{-1} +\frac 34 P\circ S_2^{-1}$, where $S_1(x)=\frac 1 4x$ and $S_2(x)=\frac 12 x +\frac 1 2$ for $x \in \mathbb R$. For this probability measure, in this paper, we investigate the optimal sets of $n$-means and the $n$th quantization errors for all positive integers $n$. The arrangement of the paper is as follows: Lemma~\ref{lemma2}, Lemma~\ref{lemma11}, and  Lemma~\ref{lemma12} give the optimal sets of $n$-means for $n=1, 2$ and $3$. Proposition~\ref{prop29}, Proposition~\ref{prop28},  Proposition~\ref{prop30}, and Proposition~\ref{prop31} give some properties about the optimal sets of $n$-means for all $n\geq 2$.  Theorem~\ref{Th1} gives the induction formula to determine the optimal sets of $n$-means and the $n$th quantization errors for all $n\geq 2$. In Section~\ref{sec4}, using the induction formula we obtain some results and observations about the optimal sets of $n$-means for $n\in \D N$.
In Section~\ref{sec5}, we have shown that the same induction formula works for the Cantor distribution $P:=\psi^2 P\circ S_1^{-1} +\psi^4 P\circ S_2^{-1}$ supported by the Cantor set generated by $S_1(x)=\frac 13x$ and $S_2(x)=\frac 13 x+\frac 23$ for all $x\in \mathbb R$, where $\psi$ is the square root of the Golden ratio $\frac 12(\sqrt 5-1)$. In fact, the same induction formula also works for the uniform Cantor distribution considered by Graf-Luschgy (see \cite {GL2}). In Section~\ref{sec6}, we give a counter example to show that the induction formula does not work for all Cantor distributions. Finally, we would like to mention that quantization for uniform Cantor distributions were investigated by several authors, for example, see \cite{GL2, K1, K2, KZ}. But, to the best of our knowledge, the work in this paper is the first advance to investigate the quantization for nonuniform Cantor distributions. The main difference with the uniform and the nonuniform distributions is that for a uniform distribution there is a closed formula for optimal quantizers for $n$-means for all $n\geq 2$ (see \cite{GL2}), but for the nonuniform distribution, considered in this paper, to obtain the optimal quantizers for $n$-means for all $n\geq 2$ a closed formula is not known yet.

\section{Preliminaries}
In this section, we give the basic preliminaries.
$\gs\tau:=\gs_1\cdots \gs_k\tau_1\cdots \tau_\ell$ represents the
concatenation of the words $\gs:=\gs_1\gs_2\cdots \gs_k$ and
$\tau:=\tau_1\tau_2\cdots \tau_\ell$ in $\{1, 2\}^*$. We say that a word $\gs$ is a \tit{predecessor} of a word $\gt$, denoted by $\gs\prec \gt$, if $\gt=\gs\gg$ for some word $\gg\in \set{1,2}^\ast$.
For $\gs=\gs_1 \gs_2 \cdots \gs_k \in \{1, 2\}^k$, let us write
$c(\gs):=\#\{i : \gs_i=1, \, 1\leq i\leq k\}$. Then
$\{J_\gs\}_{\gs \in \{1, 2\}^k}$ is the set of  $2^k$ intervals with the
length of $J_\gs$ equals $\lambda(J_\gs):= \frac{1}{4^{c(\gs)}}\frac 1 {2^{k-c(\gs)}} =\frac{1}{2^{k+c(\gs)}}$ at the $k$th level of the
Cantor construction, where $\lambda$ denotes the Lebesgue measure on $[0, 1]$. The intervals $J_{\gs1}$, $J_{\gs 2}$ into which $J_\gs$
is split up at the $(k+1)$th level are called the \tit{children} of $J_\gs$. Moreover, for any $\gs \in \{1, 2\}^*$, we have $P(J_\gs)=p_\gs=\frac{3^{|\gs|-c(\gs)}}{4^{|\gs|}}$,
and $\lambda(J_{\gs})=\frac{1}{2^{|\gs|+c(\gs)}}$.

Let us now state the following two lemmas. The proofs are routine (see \cite{GL2}).
\begin{lemma} \label{lemma1}
Let $f : \mathbb R \to \mathbb R^+$ be Borel measurable and $k\in \mathbb N$. Then
\[\int f dP=\sum_{\gs \in \{1, 2\}^k} p_\gs \int f \circ S_\gs dP.\]
\end{lemma}


\begin{lemma} \label{lemma2} Let $X$ be a real valued random variable with distribution $P$. Let $E(X)$ represent the expected value and $V:=V(X)$ represent the variance of the random variable $X$. Then,
\[E(X)=\frac 23  \text{ and } V(X)=\frac {16}{153}.\]
\end{lemma}

\begin{note}\label{note1} For any $x_0 \in \D R$, we have
$\int(x-x_0)^2 dP =V(X)+(x_0-E(X))^2$ yielding the fact that the optimal set of one-mean is the expected value and the corresponding quantization error is the variance $V$ of the random variable $X$. Since $S_1$ and $S_2$ are similarity mappings, we have $E(S_j(X))=S_j(E(X))$ for $j=1, 2$ and so, by induction, $E(S_\gs(X))=S_\gs(E(X))=S_\gs(\frac 23)$ for $\gs\in \{1, 2\}^k$, $k\geq 1$. For $\gs \in \{1, 2\}^\ast$ write
$a(\gs):=E(X : X \in J_\gs)$,  and for  $\gs,  \gt, \cdots, \gg \in \{1, 2\}^\ast$ write \[a(\gs, \gt, \cdots, \gg) :=E(X : X \in J_\gs\uu J_\gt\uu \cdots \uu J_\gg).\] Then, using Lemma~\ref{lemma1}, we have
\begin{align*}
&a(\gs)=\frac{1}{P(J_\gs)} \int_{J_\gs} x dP= \sum_{\tau\in \{1, 2\}^k} \int_{J_\gs} x\, d(P\circ S_\tau^{-1})=\int_{J_\gs} x d(P\circ S_\gs^{-1})=\int S_\gs(x) dP=E(S_\gs(X)),\\
& \te{ and similarly, } a(\gs, \gt)=\frac{1}{P(J_\gs\uu J_\gt)} \Big( P(J_\gs) S_\gs(\frac 23)+P(J_\gt) S_\gt(\frac 23)\Big).
\end{align*}
For any $a\in \D R$ and $\gs \in \set{1, 2}^\ast$, we have
\begin{equation} \label{eq1} \int_{J_\gs}(x-a)^2 dP= p_\gs \int (x -a)^2 d(P\circ S_\gs^{-1})=p_\gs \Big(s_\gs^2  V+\Big(S_\gs(\frac 23)-a\Big)^2\Big).\end{equation}
The equation~\eqref{eq1} is used to determine the quantization error.
\end{note}
 In the next section we determine the optimal sets of $n$-means and the $n$th quantization errors $V_n$ for all $n\geq 2$.

\section{Optimal sets and the error for all $n\geq 2$} \label{sec3}

In this section, we first give some lemmas and propositions that we need to deduce the theorem Theorem~\ref{Th1} which gives the induction formula to determine the optimal sets of $n$-means and the $n$th quantization errors for all $n\geq 2$.

\begin{lemma}\label{lemma11}
Let $\ga=\{a_1, a_2\}$ be an optimal set of two-means, $a_1<a_2$. Then, $a_1=a(1)=S_1(\frac 23)$,  $a_2=a(2)=S_2(\frac 2 3 )$, and the quantization error is
$V_2=\frac{13}{612}=0.0212418.$
\end{lemma}

\begin{proof} Let us first consider a two-point set $\beta$ given by $\beta:=\{a(1), a(2)\}$. Then,
\begin{align*}
\int \min_{b \in \beta}(x-b)^2  dP&=\sum_{i=1}^2 \int_{J_i} (x-a(i))^2 dP=\frac 1{64} V+\frac 3{16} V=\frac{13}{612}=0.0212418.
\end{align*}
Since $V_2$ is the quantization error for two-means, we have $V_2\leq \frac {13}{612}=0.0212418$.
Let $\ga:=\{a_1, a_2\}$ be an optimal set of two-means. Since $a_1$ and $a_2$ are the expected values of the random variable $X$ with distribution $P$ in their own Voronoi regions, we have $0\leq a_1<a_2\leq 1$. Suppose that $a_2\leq \frac{47}{64}<\frac 34=S_{22}(0)$.
Then using \eqref{eq1}, we have
\begin{align*}
V_2&=\int \min_{a \in \ga} (x-a)^2 dP\geq \int_{J_{22}} (x -\frac{47}{64})^2 dP=\frac{24921}{1114112}=0.0223685>V_2,
\end{align*}
which is a contradiction. So, we can assume that $\frac{47}{64}<a_2$. Since $a_1\geq 0$, we have $\frac 12(a_1+a_2)\geq \frac 12(0+\frac{47}{64})=\frac{47}{128}>\frac 14$, yielding the fact that the Voronoi region of $a_1$ may contain points from $J_2$, but the Voronoi region of $a_2$ does not contain any point from $J_1$ implying $E(X : X\in J_1)=a(1)=\frac 16\leq a_1$ and $E( X : X \in J_2)=\frac 56\leq a_2$. Notice that if the Voronoi region of $a_1$ does not contain any point from $J_2$, then $a_1=\frac 16$ and $a_2=\frac 56$. If $\frac{15}{32}<a_1$, then
\[V_2\geq \int_{J_1}(x-\frac{15}{32})^2 dP=\frac{5107}{208896}=0.0244476>V_2,\]
which yields a contradiction, and so $\frac 16\leq a_1\leq \frac{15}{32}$.
 We now show that the Voronoi region of $a_1$ does not contain any point from $J_2$. Notice that $\frac 16<a_1\leq\frac{15}{32}$ and $\frac 56<a_2\leq 1$ implying $\frac 12<\frac 12(a_1+a_2)\leq \frac {47}{64}<S_{22}(0)$. For the sake of contradiction, assume that the Voronoi region of $a_1$ contains points from $J_2$, and so the following two cases can arise:

Case~1. $S_{21}(1)=\frac 58\leq \frac 12(a_1+a_2)\leq \frac {47}{64}<S_{22}(0)$.

Then, $a_1=E(X : X\in J_1\uu J_{21})=\frac{29}{84}$  and $a_2=E(X : X \in J_{22})=\frac{11}{12}$ implying
\[V_2=\int_{J_1\uu J_{21}}(x-\frac {29}{84})^2 dP+\int_{J_{22}}(x-\frac {11}{12})^2 dP=\frac{415}{17136}=0.024218>V_2,\]
which gives a contradiction.


Case~2. $\frac 12=S_{21}(0)<\frac 12(a_1+a_2)<S_{21}(1)=\frac 58$.

Then, there exists a word  $\gs \in \set{1, 2}^\ast$ such that
$S_{2\gs 1}(1)\leq \frac 12(a_1+a_2)\leq S_{2\gs2}(0)$. For definiteness sake and calculation simplicity, take $\gs=1^9=111111111$, where for any positive integer $k$, by $1^k$ it is meant that the word is obtained from $k$ times concatenation of the symbol 1. Thus, we have $S_{21^{9}1}(1)\leq \frac 12(a_1+a_2)\leq S_{21^92}(0)$, and so
\begin{align*}
a_1&=E(X : X\in J_1\uu J_{21^91})=\frac{549760532483}{3298544320512}, \te{ and } \\
a_2&=E(X : X\in \mathop{\uu}_{k=1}^9J_{21^k2}\uu J_{22})=\frac{2621441}{3145728},
\end{align*}
implying
\begin{align*} &V_2=\int\min_{a\in \ga}(x-a)^2 dP=\int_{J_1\uu J_{21^91}}(x-a_1)^2 dP+\int_{\mathop{\uu}\limits_{k=1}^9J_{21^k2}\uu J_{22}}(x-a_2)^2dP\\
&=0.021241830065359477413>0.021241830065359477124=\frac {13}{612}\geq V_2,
\end{align*}
which leads to a contradiction. Similarly, we can show that a contradiction arises for any other choice of $\gs \in \set{1, 2}^\ast$ satisfying $S_{2\gs 1}(1)\leq \frac 12(a_1+a_2)\leq S_{2\gs2}(0)$. Hence, the Voronoi region of $a_1$ does not contain any point from $J_2$ yielding $a_1\leq a(1)=\frac 16$. Again, we have seen that $a_1\geq \frac 16$. Thus, we deduce that $a_1=a(1)=\frac 16$ and $a_2=a(2)=\frac 56$, and then the quantization error is $V_2=\frac {13}{612}$, which is the lemma.
\end{proof}

The technique of the proof of the following proposition is similar to Lemma 4.5 in \cite{GL2}.
\begin{prop}\label{prop29}
Let $n\geq 2$ and let $\ga_n$ be an optimal set of $n$-means such that $\ga_n$ contains points from $J_1$ and $J_2$, and $\ga_n$
does not contain any point from the open interval $(\frac 14, \frac 12)$. Further, assume that the Voronoi region of any point in $\ga_n\ii J_1$ does not contain any point from $J_2$ and the Voronoi region of any point in $\ga_n\ii J_2$ does not contain any point from $J_1$. Set $\ga_1:=\ga_n\ii J_1$, $\ga_2:=\ga_n\ii J_2$, and $j:=\te{card}(\ga_1)$. Then, $S_1^{-1}(\ga_1)$ is an optimal set of $j_1$-means and $S_2^{-1}(\ga_2)$ is an optimal set of $(n-j_1)$-means for the probability measure $P$, and
\[V_n=\frac 1{64} V_{j_1}+\frac{3}{16} V_{n-j_1}.\]
\end{prop}

We now state and prove the following lemma.
\begin{lemma} \label{lemma12} Let $\ga=\set{a_1, a_2, a_3}$ be an optimal set of three-means such that $a_1<a_2<a_3$. Then, $a_1=a(1)=S_1(\frac 23)=\frac 16$, $a_2=a(21)=S_{21}(\frac 23)=\frac 7{12}$, and $a_3=a(22)=S_{22}(\frac 23)=\frac {11}{12}$, and $V_3=\frac{55}{9792}=0.00561683.$
\end{lemma}

\begin{proof}
Let us first consider a set of three points given by $\gb:=\set{a(1), a(21), a(22)}$. The distortion error due to the set $\gb$ is given by
\[\int\min_{a\in \gb}(x-a)^2 dP=\int_{J_1}(x-a(1))^2 dP+\int_{J_{21}}(x-a(21))^2dP+\int_{J_{22}}(x-a(22))^2 dP=\frac{55}{9792}.\]
Since $V_3$ is the quantization error for three-means, $V_3\leq \frac{55}{9792}=0.00561683$. Let $\ga:=\set{a_1, a_2, a_3}$ be an optimal set of three-means. Since the points in an optimal set are the expected values in their own Voronoi regions with respect to the probability distribution $P$, we have $0\leq a_1<a_2<a_3\leq 1$. If $a_3\leq \frac {27}{32}=\frac 12(S_{221}(1)+S_{222}(0))$, then
\[V_3\geq \int_{J_{222}}(x-\frac{27}{32})^2 dP=\frac{6939}{1114112}=0.00622828>V_3,\]
which is a contradiction. So, we can assume that $\frac {27}{32}<a_3$. If $a_1>\frac 5{16}$, then
\[V_3\geq \int_{J_1}(x-\frac 5{16})^2 dP=\frac{121}{17408}=0.00695083>V_3,\]
which yields a contradiction. If $\frac 14<a_1\leq \frac 5{16}$, then $\frac 12(a_1+a_2)>\frac 12$ implying $a_2>1-a_1\geq 1-\frac 5{16}=\frac {11}{16}>\frac 58=S_{21}(1)$, and so
\[V_3\geq \int_{J_1}(x-\frac 14)^2 dP+\int_{J_{21}}(x-\frac{11}{16})^2 dP=\frac{1193}{208896}=0.00571098>V_3,\]
which leads to a contradiction. Thus, we can assume that $a_1\leq \frac 14$.
Suppose that $a_2<\frac 7{16}$. Then, $S_{21}(1)=\frac 58<\frac 12(\frac 7{16}+a(22))=\frac {65}{96}<\frac 34=S_{22}(0)$ implying
\[V_3\geq \int_{J_{21}}(x-\frac 7{16})^2 dP+\int_{J_{22}}(x-a(22))^2dP=\frac{555}{69632}=0.00797047>V_3,\]
which is a contradiction. Next, suppose that $\frac 7{16}\leq a_2<\frac 12$. Then, $\frac 12(a_1+a_2)<\frac 14$ implying $a_1<\frac 12-a_2\leq \frac 12-\frac 7{16}=\frac 1{16}$. Moreover, $S_{21}(1)<\frac 12(\frac 12+a(22))=\frac {17}{24}<S_{22}(0)$, and so
\[V_3\geq \int_{J_{12}}(x-\frac 1{16})^2 dP+\int_{J_{21}}(x-\frac 12)^2 dP+\int_{J_{22}}(x-a(22))^2dP=\frac{667}{69632}=0.00957893>V_3,\]
which leads to a contradiction. Therefore, we can assume that $\frac 12\leq a_2$, and this implies that the Voronoi region of $a_2$ does not contain any point from $J_1$. Suppose that the Voronoi region of $a_1$ contains points from $J_2$. Then, $\frac 12(a_1+a_2)>\frac 12$, and so $a_2>1-a_1\geq 1-\frac 14=\frac 34>S_{21}(1)$ implying
\[V_3\geq \int_{J_1}(x-a(1))^2 dP+\int_{J_{21}}(x-\frac 34)^2 dP=\frac{35}{4896}=0.00714869>V_3,\]
which gives a contradiction. So, we can assume that the Voronoi region of $a_1$ does not contain any point from $J_2$. This implies that $a_1=a(1)=\frac 16$. Hence, by Proposition~\ref{prop29}, $S_2^{-1}(\ga_2)$ is an optimal set of two-means, which by Lemma~\ref{lemma11}, implies that $S_2^{-1}(\ga_2)=\set{a(1), a(2)}$ yielding $\ga_2=\set{a(21), a(22)}$. Hence, $\ga=\set{a(1), a(21), a(22)}$ is an optimal set of three-means and the corresponding quantization error is given by $V_3=\frac{55}{9792}=0.00561683$. Thus, the proof of the lemma is complete.
\end{proof}

\begin{prop} \label{prop28}
Let $\ga_n$ be an optimal set of $n$-means for $n\geq 2$. Then,
\[\ga_n\ii J_1\neq \es \te{ and } \ga_n\ii J_2 \neq \es.\]
Moreover, the Voronoi region of any point in $\ga_n\ii J_1$ does not contain any point from $J_2$, and the Voronoi region of any point in $\ga_n\ii J_2$ does not contain any point from $J_1$.
\end{prop}
\begin{proof} By Lemma~\ref{lemma11} and Lemma~\ref{lemma12}, the proposition is true for $n=2$ and $n=3$. We now prove that the proposition is true for $n\geq 4$. Let $\ga_n:=\set{0\leq a_1<a_2<\cdots<a_n\leq 1}$ be an optimal set of $n$-means for $n\geq 4$. Consider the set of four points given by $\gb:=\set{a(1), a(21), a(221), a(222)}$. Then,
\begin{align*}
\int\min_{a \in \gb}(x-a)^2 dP& =\int_{J_1}(x-a(1))^2 dP+\int_{J_{21}}(x-a(21))^2 dP+\int_{J_{221}}(x-a(221))^2 dP\\&+\int_{J_{222}}(x-a(222)^2dP=\frac{421}{156672}=0.00268714.
\end{align*}
Since $V_n$ is the quantization error for $n$-means for $n\geq 4$, we have $V_n\leq V_4\leq 0.00268714$. If $a_1\geq \frac 14$, then
\[V_n\geq \int_{J_1}(x-\frac 14)^2 dP=\frac{11}{3264}=0.0033701>V_n,\]
which leads to a contradiction. So, we can assume that $a_1<\frac 14$. If $a_n\leq \frac 12$, then
\[V_n\geq \int_{J_2}(x-\frac 12)^2 dP=\frac {7}{68}=0.102941>V_n,\]
which yields a contradiction, and so $\frac 12<a_n$. This completes the proof of the first part of the proposition. To complete the proof of the proposition, let $j:=\max \set{i : a_i\leq \frac 14}$. Then, $a_j\leq \frac 14$. Suppose that the Voronoi region of $a_j$ contains points from $J_2$. Then, $\frac 12(a_j+a_{j+1})>\frac 12$ implying $a_{j+1}>1-a_j\geq 1-\frac 14=\frac 34=S_{22}(0)$, and so
\[V_n\geq \int_{J_{21}}\min_{a\in \ga_n}(x-a)^2 dP=\int_{J_{21}}\min_{a\in \set{a_{j}, a_{j+1}}} (x-a)^2 dP\geq \int_{J_{21}}(x-\frac 34)^2 dP=\frac 3{544}=0.00551471>V_n,\]
which is a contradiction. Next, let $k=\min\set{i : a_i\geq \frac 12}$ implying $\frac 12\leq a_k$. Assume that the Voronoi region of $a_k$ contains points from $J_1$. Then, $\frac 12(a_{k-1}+a_k)<\frac 14$ implying $a_{k-1}<\frac 12-a_k\leq \frac 12-\frac 12=0$, which is a contradiction as $0\leq a_1\leq a_2<\cdots <a_n\leq 1$. Thus, the proof of the proposition is complete.
\end{proof}

Let us now state the following proposition. Due to technicality, we do not show the proof of it in the paper.
\begin{prop} \label{prop30}
Let $\ga_n$ be an optimal set of $n$-means for any $n\geq 2$. Then, $\ga_n$ does not contain any point from the open interval $(\frac 14, \frac 12)$.
\end{prop}

The following proposition plays an important role in the paper.

\begin{prop}\label{prop31}
Let $\ga_n$ be an optimal set of $n$-means for $n\geq 2$. Then, for $c\in \ga_n$, we have $c=a(\gt)$ for some $\gt \in \set{1, 2}^\ast$.
\end{prop}
\begin{proof}
Let $\ga_n$ be an optimal set of $n$-means for $n\geq 2$ and $c\in \ga_n$. By Proposition~\ref{prop28}, and Proposition~\ref{prop30}, we see that  either $c\in \ga_n\ii J_1$ or $c\in \ga_n\ii J_2$. Without any loss of generality, we can assume that  $c\in \ga_n\ii J_1$. If $\te{card}(\ga_n\ii J_1)=1$, then by Proposition~\ref{prop29}, $S_1^{-1}(\ga_n\ii J_1)$ is an optimal set of one-mean yielding $c=S_1(\frac 23)=a(1)$. Assume that $\te{card}(\ga_n\ii J_1)\geq 2$. Then, as similarity mappings preserve the ratio of the distances of a point from any other two points, using Proposition~\ref{prop29} again, we have $(\ga_n\ii J_1)\ii J_{11}=\ga_n\ii J_{11}\neq \es$ and $(\ga_n\ii J_1)\ii J_{12}=\ga_n\ii J_{12}\neq \es$, and $\ga_n\ii J_1$ does not contain any point from the open interval $(S_{11}(1), S_{12}(0))$ yielding the fact that $c\in (\ga_n\ii J_{11})\uu (\ga_n\ii J_{12})$. Without any loss of generality, assume that $c\in \ga_n\ii J_{11}$. If $\te{card}(\ga_n\ii J_{11})=1$, by Proposition~\ref{prop29} as before, we see that $S_{11}^{-1}(\ga_n\ii J_{11})$ is an optimal set of one-mean implying $c=S_{11}(\frac 23)$. If $\te{card}(\ga_n\ii J_{11})\geq 2$, then proceeding inductively in the similar way, we can find a word $\gt\in \set{1, 2}^\ast$ with $11\prec \gt$, such that $c\in \ga_n\ii J_\gt$ and $\te{card}(\ga_n\ii J_\gt)=1$, and then $S_\gt^{-1}(\ga_n\ii J_\gt)$ being an optimal set of one-mean for $P$, we have $c=S_\gt(\frac 23)=a(\gt)$. Thus, the proof of the proposition is yielded.
\end{proof}


We need the following lemma to prove the main theorem Theorem~\ref{Th1}.
\begin{lemma}\label{lemma16}
Let  $\gs, \tau \in \{1, 2\}^\ast$. Then
\begin{align*} P(J_{\gs1})(\lambda(J_{\gs 1}))^2&+P(J_{\gs2})(\lambda(J_{\gs 2}))^2+P(J_{\tau})(\lambda(J_{\tau}))^2\\
&< P(J_{\gs})(\lambda(J_{\gs}))^2+P(J_{\tau1})(\lambda(J_{\tau 1}))^2+P(J_{\tau2})(\lambda(J_{\tau 2}))^2
\end{align*}
if and only if  $P(J_\gs)(\lambda(J_\gs))^2 > P(J_\tau)(\lambda(J_\tau))^2$.
\end{lemma}
\begin{proof}
For any $\eta \in \{1, 2\}^*$, we have
$
P(J_{\eta 1})=\frac 1 4 P(J_\eta)$, $P(J_{\eta 2}) =\frac 3 4 P(J_\eta)$, $\lambda(J_{\eta 1})=\frac 1 4 \lambda(J_\eta)$,  \te{ and } $\lambda(J_{\eta 2}) =\frac 1 2 \lambda(J_\eta)$.
Then, for $\gs, \tau \in \{1, 2\}^\ast$,
\begin{align*}
&\left(P(J_{\gs1})(\lambda(J_{\gs 1}))^2  +P(J_{\gs2})(\lambda(J_{\gs 2}))^2+P(J_{\tau})(\lambda(J_{\tau}))^2 \right)\\
&\qquad \qquad \qquad  -\left(P(J_{\gs})(\lambda(J_{\gs}))^2+P(J_{\tau1})(\lambda(J_{\tau 1}))^2+P(J_{\tau2})(\lambda(J_{\tau 2}))^2\right)\\
&=\left(\frac 1{64} P(J_{\gs})(\lambda(J_{\gs}))^2  + \frac 3{16} P(J_{\gs})(\lambda(J_{\gs}))^2+P(J_{\tau})(\lambda(J_{\tau}))^2 \right)\\
&\qquad \qquad  \qquad -\left(P(J_{\gs})(\lambda(J_{\gs}))^2+\frac 1{64} P(J_{\tau})(\lambda(J_{\tau}))^2+ \frac{3}{16}P(J_{\tau})(\lambda(J_{\tau}))^2\right)\\
&=\frac 1 {64} \left(P(J_{\gs})(\lambda(J_{\gs}))^2-P(J_\tau)(\lambda(J_\tau))^2\right)+\frac 3 {16} \left(P(J_{\gs})(\lambda(J_{\gs}))^2-P(J_\tau)(\lambda(J_\tau))^2\right)\\
&\qquad \qquad \qquad -\left(P(J_{\gs})(\lambda(J_{\gs}))^2-P(J_\tau)(\lambda(J_\tau))^2\right)\\
&=-\frac{51}{64}\left(P(J_{\gs})(\lambda(J_{\gs}))^2-P(J_\tau)(\lambda(J_\tau))^2\right),
\end{align*}
and thus, the lemma follows.
\end{proof}

Due to Proposition~\ref{prop31} and Lemma~\ref{lemma16}, we are now ready to state and prove the following theorem, which gives the induction formula to determine the optimal sets of $n$-means and the $n$th quantization errors for all $n\geq 2$.
\begin{theorem} \label{Th1} Let $n\in \D N$ and $n\geq 2$. Let $\ga_n$ be an optimal set of $n$-means, i.e., $\ga_n \in\C C_n:= \mathcal{C}_n(P)$. Set
$
O_n(\ga_n):=\{\gs \in \{1, 2\}^* : S_\gs(\frac 23 ) \in \ga_n\}$, and
\[\hat O_n(\ga_n):=\{\tau \in O_n(\ga_n) : P(J_{\tau})\left(\lambda(J_\tau)\right)^2 \geq P(J_{\gs})\left(\lambda(J_\gs)\right)^2 \text{ for all } \gs \in O_n(\ga_n)\}.\]
 Take any $\tau \in \hat O_n(\ga_n)$.
Then, $\ga_{n+1}(\tau):=\{S_{\gs}(\frac 23 ): \gs \in (O_n(\ga_n)\setminus \{\tau\})\} \cup \{S_{\tau 1}(\frac 23 ), \,S_{\tau 2}(\frac 23 )\}$ is an optimal set of $(n+1)$-means for $P$, and the number
of such sets is given by
\[\text{card}\Big(\bigcup_{\ga_n \in \mathcal{C}_n}\{\ga_{n+1}(\tau) : \tau \in \hat O_n(\ga_n)\}\Big).\]
Moreover, the $n$th quantization error is given by
\[V_n=\sum_{\gs \in O_n(\ga_n)} P(J_{\gs})(\lambda(J_{\gs}))^2\,V=\sum_{\gs\in O_n} \frac{3^{|\gs|-c(\gs)}}{2^{4|\gs|+2c(\gs)}}\,V.\]
\end{theorem}

\begin{proof}
For $n\geq 2$, let $\ga_n$ be an optimal set of $n$-means for $P$. Then, $\ga_{n+1}$ is an optimal set of $(n+1)$-means. Let $\te{card}(\ga_{n+1}\ii J_1)=n_1$ and $\te{card}(\ga_{n+1}\ii J_2)=n_2$. Then, by Proposition~\ref{prop29}, Proposition~\ref{prop28}, and Proposition~\ref{prop30}, $S_1^{-1}(\ga_{n+1}\ii J_1)$ is an optimal set of $n_1$-means implying $\ga_{n+1}\ii J_1=S_1(\ga_{n_1})$.  Similarly, $\ga_{n+1}\ii J_2=S_2(\ga_{n_2})$. Thus, we see that $\ga_{n+1}=S_1(\ga_{n_1})\uu S_2(\ga_{n_2})$. Again, by Proposition~\ref{prop31}, we see that for any $c\in \ga_n$, $c=a(\gs)$ for some $\gs\in \set{1, 2}^\ast$. Hence, we can conclude that
\begin{equation} \label{eq0001} \ga_{n+1}\sci \uu\set{J_\gs : a(\gs) \in \ga_n}\sci \uu \set{J_\gs : a(\gs)\in \ga_{n-1}}\sci \cdots \sci J_1\uu J_{21}\uu J_{22}\sci J_1\uu J_2. \end{equation}
Write $
O_n(\ga_n):=\{\gs \in \{1, 2\}^* : S_\gs(\frac 23 ) \in \ga_n\}$, and
$\hat O_n(\ga_n):=\{\tau \in O_n(\ga_n) : P(J_{\tau})\left(\lambda(J_\tau)\right)^2 \geq P(J_{\gs})\left(\lambda(J_\gs)\right)^2 \text{ for all } \gs \in O_n(\ga_n)\}.$
If $\tau \not \in \hat O_n(\ga_n)$, i.e., if  $\tau \in O_n(\ga_n)\setminus \hat O_n(\ga_n)$, then by Lemma~\ref{lemma16}, the error \[\int\min_{\gs\in(O_n(\ga_n)\setminus \{\tau\})\cup \{\tau 1, \tau 2\}}(x-S_\gs(\frac 2 3))^2 dP\] obtained in this case is strictly greater than the corresponding error obtained in the case where $\tau\in \hat O_n(\ga_n)$. Hence, by the relation~\eqref{eq0001}, we can say that for any $\tau \in \hat O_n(\ga_n)$, the set $\ga_{n+1}(\tau)=\{S_{\gs}(\frac 23 ): \gs \in (O_n(\ga_n)\setminus \{\tau\})\} \cup \{S_{\tau 1}(\frac 23 ), \,S_{\tau 2}(\frac 23 )\}$ is an optimal set of $(n+1)$-means for $P$, and the number
of such sets equals $\text{card}\Big(\bigcup_{\ga_n \in \mathcal{C}_n(P)}\{\ga_{n+1}(\tau) : \tau \in \hat O_n(\ga_n)\}\Big)$. Moreover, the $n$th quantization error is given by
\begin{align*}
V_n&=\int\min_{\gs \in O_n(\ga_n)}\Big(x-S_\gs(\frac 2 3 )\Big)^2dP=\sum_{\gs \in O_n(\ga_n)} \int_{J_\gs} \Big(x-S_\gs(\frac 2 3 )\Big)^2dP\\
&=\sum_{\gs \in O_n(\ga_n)} P(J_{\gs})(\lambda(J_{\gs}))^2\,V=\sum_{\gs \in O_n(\ga_n)}\frac{3^{|\gs|-c(\gs)}}{4^{|\gs|}}\Big(\frac 1{2^{|\gs|+c(\gs)}}\Big)^2\,V=\sum_{\gs\in O_n(\ga_n)} \frac{3^{|\gs|-c(\gs)}}{2^{4|\gs|+2c(\gs)}}\,V.
\end{align*}
Thus, the proof of the theorem is complete.
\end{proof}

\begin{remark}
For the probability distribution supported by the nonuniform Cantor distribution, considered in this paper, to obtain an optimal set of $(n+1)$-means one needs to know an optimal set of $n$-means. A closed formula is not known yet. Further investigation in this direction is still awaiting.
\end{remark}

Using the induction formula given by Theorem~\ref{Th1}, we obtain some results and observations about the optimal sets of $n$-means, which are given in the following section.

\begin{figure}
\centerline{\includegraphics[ ]{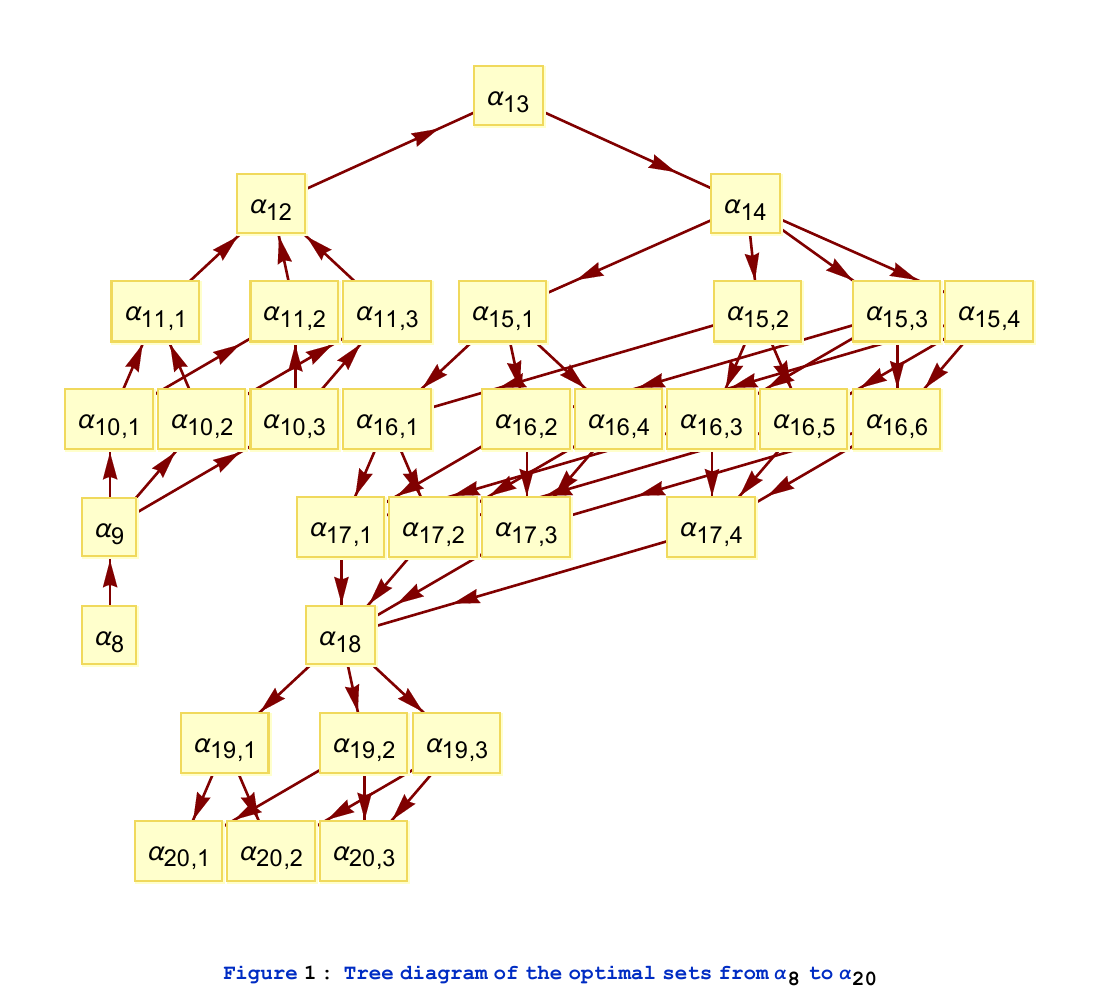}}
\end{figure}

\section{Some results and observations}  \label{sec4}

 Let $\ga_n$ be an optimal set of $n$-means, i.e., $\ga_n \in \C C_n$, and then for any $a\in \ga_n$, we have $a:=a(\gs)=S_\gs(\frac 23)$ for some $\gs:=\gs_1\gs_2\cdots\gs_k \in \set{1, 2}^k$, $k\geq 1$. Moreover, $a$ is the conditional expectation of the random variable $X$ given that $X$ is in $J_\gs$, i.e., $a=S_\gs(\frac 2 3)=E(X : X \in J_\gs)$.
If $\te{card}(\C C_n)=k$ and  $\te{card}(\C C_{n+1})=m$, then either $1\leq k\leq m$, or $1\leq m\leq k$ (see Table~\ref{tab1}). Moreover, by Theorem~\ref{Th1}, we see that an optimal set at stage $n$ can generate multiple distinct optimal sets at stage $n+1$, and multiple distinct optimal sets at stage $n$ can produce one common optimal set at stage $n+1$; for example from Table~\ref{tab1}, we see that the number of $\ga_9=1$, the number of $\ga_{10}=3$, the number of $\ga_{11}=3$, and the number of $\ga_{12}=1$. By $\ga_{n, i} \rightarrow \ga_{n+1, j}$, it is meant that the optimal set $\ga_{n+1, j}$ at stage $n+1$ is produced from the optimal set $\ga_{n, i}$ at stage $n$, similar is the meaning for the notations $\ga_n\rightarrow \ga_{n+1, j}$, or $\ga_{n, i} \rightarrow \ga_{n+1}$, for example from Figure~1:
 \begin{align*} &\left\{\ga _9\to \ga _{10,1},\ga _9\to \ga _{10,2}, \ga _9\to \ga _{10,3}\right\},\\
 &\left\{\left\{\ga _{10,1}\to \ga _{11,1},\ga _{10,1}\to \ga _{11,2}\right\}, \left\{\ga _{10,2}\to \ga _{11,1},\ga _{10,2}\to \ga _{11,3}\right\},\left\{\ga _{10,3}\to \ga _{11,2},\ga _{10,3}\to \ga _{11,3}\right\}\right\},\\
  &\left\{\ga _{11,1}\to \ga _{12},\ga _{11,2}\to \ga _{12},\ga _{11,3}\to \ga _{12}\right\}.
 \end{align*}
Moreover, we see that
\begin{align*}
\ga_9&=\{a(11),a(121),a(122),a(211),a(212),a(221),a(2221), a(22221),a(22222)\} \\
& \qquad \te{ with } V_9=\frac{9805}{40108032}=0.000244465;\\
\ga_{10, 1}&=\{a(11),a(121),a(122),a(211),a(212),a(2211),a(2212),a(2221), a(22221),a(22222)\};\\
\ga_{10, 2}&=\{a(11),a(121),a(122),a(211),a(221),a(2121),a(2122),a(2221), a(22221),a(22222)\};\\
\ga_{10, 3}&=\{a(11),a(121),a(211),a(212),a(221),a(1221),a(1222),a(2221), a(22221),a(22222)\} \\
& \qquad   \te{ with } V_{10}=\frac{7969}{40108032}=0.000198688;
\end{align*}
\begin{align*}
\ga_{11, 1}&=\{a(11),a(121),a(122),a(211),a(2121),a(2122),a(2211),a(2212), a(2221),a(22221), \\
& \qquad a(22222)\};\\
\ga_{11, 2}&=\{a(11),a(121),a(211),a(212),a(1221),a(1222),a(2211),a(2212), a(2221),a(22221), \\
& \qquad a(22222)\};\\
\ga_{11, 3}&=\{a(11),a(121),a(211),a(221),a(1221),a(1222),a(2121),a(2122), a(2221),a(22221), \\
& \qquad a(22222)\}\te{ with } V_{11}=\frac{6133}{40108032}=0.000152912;\\
\ga_{12}&=\{a(11),a(121),a(211),a(1221),a(1222),a(2121),a(2122),a(2211), a(2212),a(2221), \\
& \qquad a(22221),a(22222)\}\te{ with } V_{12}=\frac{4297}{40108032}=0.000107136; \\
\ga_{13}&=\{a(111),a(112), a(121),a(211),a(1221),a(1222),a(2121),a(2122), a(2211), a(2212), \\
& \qquad a(2221),a(22221),a(22222)\}\\
& \text{ with } V_{13}=\frac{3481}{40108032}=0.0000867906;
\end{align*}
and so on.
\begin{table}
\begin{center}
\begin{tabular}{ |c|c||c|c|| c|c||c|c|c||c|c||c|c}
 \hline
$n$ & $\te{card}(\C C_n) $ & $n$ & $\te{card}(\C C_n) $  & $n$ & $\te{card}(\C C_n)  $  & $n$ & $\te{card}(\C C_n)$   & $n$ & $\te{card}(\C C_n)$ & $n$ & $\te{card}(\C C_n)$\\
 \hline
5 & 1 & 18 &  1 &  31& 15 & 44 & 120& 57 & 7& 70 & 6435 \\6 & 1 &  19 & 3 & 32 & 6 & 45& 210 & 58 & 21& 71 & 6435\\7 & 2 & 20  & 3 & 33&  1& 46 & 252 & 59 & 35& 72 & 5005\\8 & 1 &  21 &  1& 34 & 1& 47& 210 & 60& 35 & 73 & 3003\\9 & 1 & 22&  1 & 35 &  1& 48 & 120& 61 & 21 & 74 & 1365\\10 & 3 & 23 & 5 & 36 & 6 & 49 & 45& 62 & 7& 75 & 455 \\11 & 3 & 24 &  10 & 37 & 15& 50  & 10& 63 & 1 & 76 & 105 \\12 & 1 & 25 & 10 & 38 & 20 & 51 & 1 & 64& 15& 77 & 15\\13 & 1 & 26 & 5 & 39& 15& 52 & 1& 65 & 105& 78 & 1 \\14 & 1 & 27& 1 & 40 & 6& 53 & 4& 66  & 455 & 79 & 1\\15 & 4 & 28&  6 & 41 &1 & 54 & 6& 67 & 1365 & 80 & 10  \\16 & 6 &29 & 15 & 42 & 10 & 55 & 4& 68 & 3003& 81 & 45 \\17 & 4 &30 & 20 & 43&  45 & 56 & 1& 69 & 5005 & 82 & 120\\
 \hline
\end{tabular}
 \end{center}
 \
\caption{Number of $\ga_n$ in the range $5\leq n\leq 82$.}
    \label{tab1}
\end{table}

\section{Quantizers for nonuniform Cantor distributions with Golden ratio} \label{sec5}
Let $\psi$ be the square root of the Golden ratio $\frac 12 (\sqrt 5-1)$, i.e., $\psi=\sqrt{\frac 12 (\sqrt 5-1)}$. Then, $\psi^2+\psi^4=1$. In this section, take $P:=\psi^2 P\circ S_1^{-1}+\psi^4 P\circ S_2^{-1}$ and $S_1(x)=\frac 13x$, and $S_2(x)=\frac 13x+\frac 23$ for $x\in \D R$. Then, $P$ is a nonuniform Cantor distribution supported by the classical Cantor set $C$ generated by $S_1$ and $S_2$ associated with the probability vector $(\psi^2, \psi^4)$. For this probability measure in this section, we investigate the optimal sets of $n$-means and the $n$th quantization errors for all $n\geq 2$.

Proceeding in the similar way as Lemma~\ref{lemma2}, the following lemma can be proved.

\begin{lemma} \label{lemma211} Let $X$ be a real valued random variable with distribution $P$. Let $E(X)$ represent the expected value and $V:=V(X)$ represent the variance of the random variable $X$. Then,
\[E(X)=\psi^4 \text{ and } V(X)=\frac{1}{2} (\psi^2-\psi^4).\]
\end{lemma}

The following note is similar to Note~\ref{note1}.
\begin{note}\label{note11} In this section, we use the same notations as defined in Section~2 just by substituting $s_1=s_2=\frac 13$, $p_1=\psi^2$ and $p_2=\psi^4$. Then, for $\gs,  \gt \in \{1, 2\}^\ast$, we have
\begin{align*}
&a(\gs)=S_\gs(\psi^4)  \te{ and } a(\gs, \gt)=\frac{1}{P(J_\gs\uu J_\gt)} \Big( P(J_\gs) S_\gs(\psi^4)+P(J_\gt) S_\gt(\psi^4)\Big).
\end{align*}
For any $a\in \D R$ and $\gs \in \set{1, 2}^\ast$, we have
\begin{equation} \label{eq11} \int_{J_\gs}(x-a)^2 dP= p_\gs \int (x -a)^2 d(P\circ S_\gs^{-1})=p_\gs \Big(s_\gs^2  V+\Big(S_\gs(\psi^4)-a\Big)^2\Big).\end{equation}
\end{note}

Let us now state the following lemmas and propositions, the technique of the proofs are similar to the similar lemmas and propositions in Section~\ref{sec3}.
\begin{lemma}\label{lemma111}
Let $\ga=\{a_1, a_2\}$ be an optimal set of two-means, $a_1<a_2$. Then, $a_1=a(1)=S_1(\psi^4)$,  $a_2=a(2)=S_2(\psi^4)$, and the quantization error is
$V_2=\frac 19 V.$
\end{lemma}

\begin{prop} \label{prop51}
Let $\ga_n$ be an optimal set of $n$-means for $n\geq 2$. Then, $\ga_n\ii J_1\neq \es$, $\ga_n \ii J_2\neq\es$, and $\ga_n$ does not contain any point from the open interval $(\frac 13, \frac 23)$.
\end{prop}

\begin{prop}\label{prop52}
Let $\ga_n$ be an optimal set of $n$-means for $n\geq 2$. Set $\ga_1:=\ga_n\ii J_1$, $\ga_2:=\ga_n\ii J_2$, and $j:=\te{card}(\ga_1)$. Then, $S_1^{-1}(\ga_1)$ is an optimal set of $j_1$-means and $S_2^{-1}(\ga_2)$ is an optimal set of $(n-j_1)$-means for the probability measure $P$, and
\[V_n=\frac 1{9}(\psi^2 V_{j_1}+\psi^4 V_{n-j_1}\Big).\]
\end{prop}


\begin{lemma}
Let $\ga$ be an optimal set of three-means. Then, $\ga=\set{a(11), a(12), a(2)}$, and the corresponding quantization error is $V_3=\frac 1{9^2}(1+8\psi^4)V$.

\end{lemma}

\begin{remark} For the uniform Cantor distribution considered by Graf-Luschgy in \cite{GL2}, there exist two different optimal sets of three-means. But, for the nonuniform Cantor distribution considered in this section, the optimal set of three-means is unique.
\end{remark}

For $\gs \in \set{1, 2}^\ast$, set $E(\gs):=\int_{J_\gs}(x-a(\gs))^2 dP$. Then, $E(\gs)$ represent the distortion error due to the point $a(\gs)$ in its own Voronoi region. Let $\gl(J_\gs)$ and $c(\gs)$ be defined as in Section~2. Then,
\begin{equation} \label{eq555} E(\gs)=P(J_\gs)\gl(J_\gs)^2 V=\frac 1{9^{|\gs|}} (\psi^2)^{ c(\gs)} (\psi^4)^{|\gs|- c(\gs)} V=\frac 1{9^{|\gs|}} \psi^{ 4|\gs|-2c(\gs)}V.\end{equation}

We now prove the following lemma.
\begin{lemma} \label{lemma456}
Let $\gs, \gt\in \set{1, 2}^\ast$. Then,  $E(\gs1)+E(\gs2)+E(\gt)<E(\gs)+E(\gt1)+E(\gt2)$ if and only if $E(\gs)>E(\gt)$; and   $E(\gs1)+E(\gs2)+E(\gt)=E(\gs)+E(\gt1)+E(\gt2)$ if and only if $E(\gs)=E(\gt)$.
\end{lemma}

\begin{proof} By \eqref{eq555}, we have
$E(\gs1)=\frac 19 \frac 1{9^{|\gs|}} \psi^{ 4|\gs|-2c(\gs)-2}V=\frac 1 9 \frac 1{\psi^2} E(\gs),$
and similarly, $E(\gs2)=\frac 19 E(\gs)$, $E(\gt 1)=\frac 19 \frac 1{\psi^2} E(\gt)$, and $E(\gt2)=\frac 19 E(\gt)$. Thus,
\begin{align*}
(E(\gs1)+E(\gs2)+E(\gt)) -(E(\gs)+E(\gt1)+E(\gt2))=\frac {1-8 \psi^2}{9\psi^2} (E(\gs)-E(\gt)).
\end{align*}
Since $\frac {1-8 \psi^2}{9\psi^2}<0$, the lemma is yielded.
\end{proof}


\begin{figure}
\centerline{\includegraphics[ ]{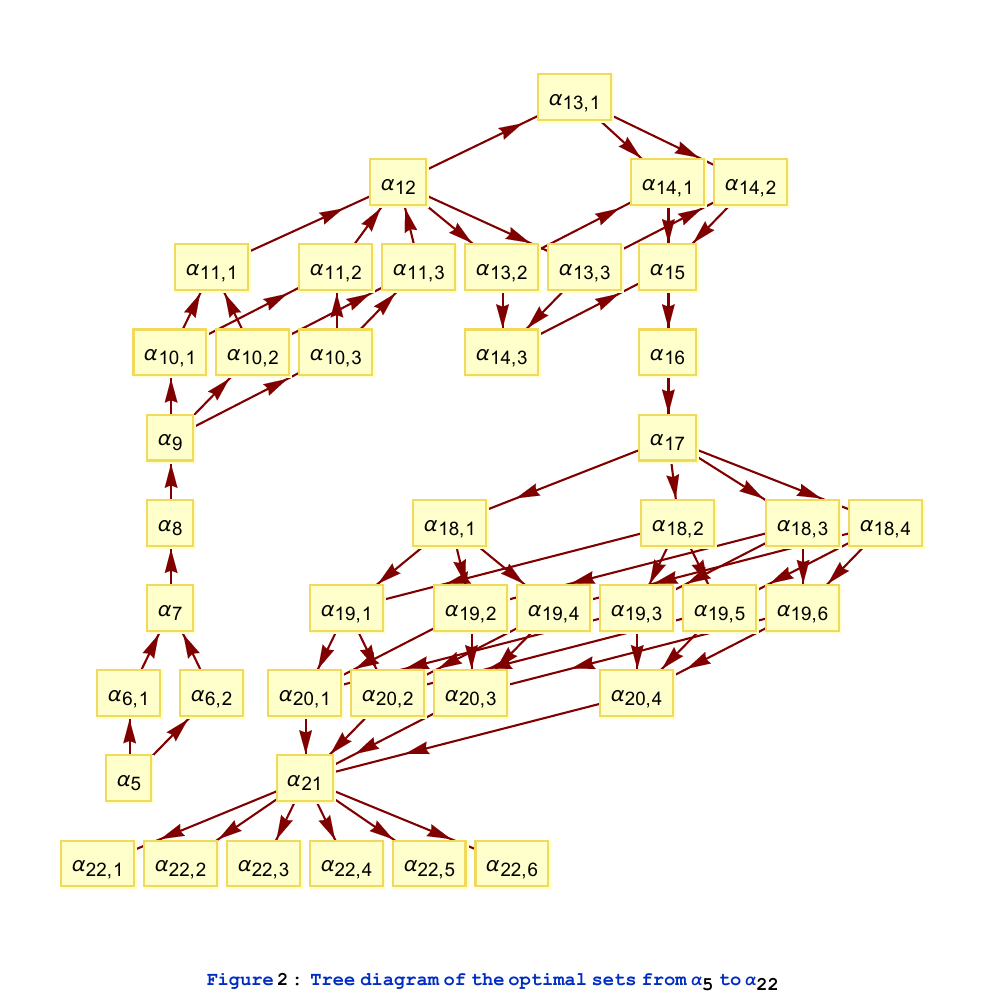}}
\end{figure}

\begin{prop} Theorem~\ref{Th1} gives the induction formula to determine the optimal sets of $n$-means and the $n$th quantization errors for the probability measure $P=\psi^2 P\circ S_1^{-1}+\psi^4 P\circ S_2^{-1}$ for all $n\geq 2$.
\end{prop}

\begin{proof}
Notice that Proposition~\ref{prop31},  and Lemma~\ref{lemma16} are also true for the probability measure $P=\psi^2 P\circ S_1^{-1}+\psi^4 P\circ S_2^{-1}$ considered in this section. Thus, if $\ga_n$ is an optimal set of $n$-means for some $n\geq 2$, then for $a(\gs), a(\gt) \in \ga_n$, by Lemma~\ref{lemma456}, if $E(\gs)>E(\gt)$, we have $E(\gs1)+E(\gs2)+E(\gt)<E(\gs)+E(\gt1)+E(\gt2)$, and if $E(\gs)=E(\gt)$, we have $E(\gs1)+E(\gs2)+E(\gt)=E(\gs)+E(\gt1)+E(\gt2)$. Hence, the induction formula given by Theorem~\ref{Th1} also determines the optimal sets of $n$-means and the $n$th quantization errors for the probability measure $P=\psi^2 P\circ S_1^{-1}+\psi^4 P\circ S_2^{-1}$ for all $n\geq 2$.
\end{proof}
\begin{remark} Using the induction formula, and the similar notations as described in Section~\ref{sec4},  we obtain a tree diagram of the optimal sets of $n$-means for $5\leq n\leq 22$ for the probability distribution $P=\psi^2 P\circ S_1^{-1}+\psi^4 P\circ S_2^{-1}$,  which is given by Figure~2. The induction formula also works for the Cantor distribution considered by Graf-Luschgy (see \cite{GL2}). To obtain the optimal sets of $n$-means for all $n\geq 2$ such an induction formula for any Cantor distribution does not always work. In the next section, we give a counter example in this direction.
\end{remark}

\section{The induction formula does not work for all Cantor distributions} \label{sec6}

In this section, we consider the Cantor distribution $P=\frac 12 P\circ S_1^{-1}+\frac 12 P\circ S_2^{-1}$, which has support the Cantor set generated by the two contractive similarity mappings given by $S_1(x)=\frac 7{16} x$ and $S_2(x)=\frac 7{16}x+\frac 9{16}$ for all $x\in \D R$. We keep the same notations as given in the previous sections. Let $X$ be the random variable with distribution $P$, then we have $E(X)=\frac 12$ and $V:=V(X)=\frac{9}{92}$.
Let us now state the following lemma. The proof is not difficult to see.

\begin{lemma}\label{lemma112}
Let $P=\frac 12 P\circ S_1^{-1}+\frac 12 P\circ S_2^{-1}$. Then, the set $\set{a_1, a_2}$, where $a_1=a(1)=S_1(\frac 12)$ and $a_2=a(2)=S_2(\frac 12)$, forms an optimal set of two-means with quantization error $V_2=\frac{441}{23552}$.
\end{lemma}

Let us now give the following lemma.

\begin{lemma}\label{lemma431} The set $\set{a(11), a(12), a(2)}$ does not form an optimal set of three-means for the probability distribution $P$ given in this section.
\end{lemma}

\begin{proof} Let $V_{3,1}$ be the distortion error due to the set $\set{a(11), a(12), a(2)}$. Then,
\[V_{3,1}= \int_{J_{11}}(x-a(11))^2 dP+\int_{J_{12}}(x-a(12))^2 dP+\int_{J_2}(x-a(2))^2 dP=\frac{134505}{12058624}=0.0111543.\]
We have $a(11, 121, 1221)=E(X : X \in J_{11}\uu J_{121}\uu J_{1221})=\frac{24919}{131072}=0.190117$. Similarly, $a(1222, 21)=\frac{400031}{655360}=0.610399$, and $a(22)=\frac{463}{512}=0.904297$.
Notice that
\begin{align*} S_{1221}(1)=0.390396&<\frac{1}{2}(a(11, 121, 1221)+ a(1222, 21))=0.400258<S_{1222}(0)=0.400864,\\
S_{21}(1)=0.753906&<\frac{1}{2}(a(1222, 21)+a(22))=0.757348<S_{22}(0)=0.808594.
\end{align*}
Thus, $P$-almost surely the set $\set{a(11, 121, 1221), a(1222, 21), a(22)}$ forms a Voronoi partition of the Cantor set $C$ generated by $S_1$ and $S_2$. Moreover, $a(11, 121, 1221)$, $a(1222, 21)$ and $a(22)$ are the expected values of the random variable $X$ in their own Voronoi regions. Let $V_{3, 2}$ be the distortion error due to the set $\set{a(11, 121, 1221), a(1222, 21), a(22)}$ with respect to the probability distribution $P$. Then,
\[V_{3,2}=\int_{J_{11}\uu J_{121}\uu J_{1221}}(x-\frac{24919}{131072})^2 dP+\int_{J_{1222}\uu J_{21}}(x-\frac{400031}{655360})^2 dP+\int_{J_{22}}(x-\frac{463}{512})^2 dP,\]
implying $V_{3,2}=\frac{88046975853}{7902739824640}=0.0111413$. Since $V_{3,2}<V_{3,1}$, the set $\set{a(11), a(12), a(2)}$ does not form an optimal set of three-means, which completes the proof of the lemma.
\end{proof}

We now prove the following proposition, which is the main result in this section.

\begin{prop} \label{prop444} Let $P$ be the cantor distribution considered in this section. Then, the induction formula given by Theorem~\ref{Th1} does not give the optimal sets of $n$-means for all $n\geq 2$.
\end{prop}
\begin{proof} For the sake of contradiction, assume that the induction formula given by Theorem~\ref{Th1} gives the optimal sets of $n$-means for all $n\geq 2$. By Lemma~\ref{lemma112}, the set $\set{a(1), a(2)}$ forms an optimal set of two-means for the Cantor distribution $P$. Using the notations given in Theorem~\ref{Th1}, we have
$O_2(\ga_2)=\hat O_2(\ga_2)=\set{1, 2}$. Take $\gt=1$. Then, by Theorem~\ref{Th1}, the set $\ga_3(\gt)=\set{a(11), a(12), a(2)}$ is an optimal set of three-means for the probability distribution $P$, which by Lemma~\ref{lemma431} gives a contradiction. Hence, the induction formula given by Theorem~\ref{Th1} does not give the optimal sets of $n$-means for all $n\geq 2$.
\end{proof}

\noindent \textbf{Acknowledgements.}
In this paper I would like to express sincere thanks and gratitude to my Master's thesis advisor Professor S.N. Lahiri.

\end{document}